\pdfoutput=1
\documentclass[conference, a4paper]{IEEEtran}
\usepackage{cite}
\usepackage{indentfirst}
\usepackage{graphicx}
\usepackage{changepage}
\usepackage{stfloats}
\usepackage{amsfonts,amssymb}
\usepackage{bm}
\usepackage{algorithm}
\usepackage{algorithmic}
\usepackage{amsmath}
\usepackage{amsmath,amssymb,amsfonts}
\usepackage{times}
\usepackage{url}
\usepackage{cite}
\usepackage{textcomp}
\usepackage{xcolor}
\usepackage{color}
\usepackage{setspace}
\usepackage{enumitem}
\usepackage{float}
\usepackage{diagbox}
\usepackage[T1]{fontenc}
\usepackage[utf8]{inputenc}
\usepackage{threeparttable}
\usepackage{booktabs}
\usepackage{footnote}
\usepackage{graphicx}
\usepackage{pifont}
\usepackage{subfigure}
\usepackage{mathrsfs}
\allowdisplaybreaks[4]
\newenvironment{proof}{{\indent  \it Proof:}}{\hfill $\blacksquare$}

\begin{document}
	
	\title{Cooperative Sensing and Communication for ISAC Networks: Performance Analysis and Optimization}
	
	\author{
		\IEEEauthorblockN{Kaitao Meng\IEEEauthorrefmark{1} and Christos Masouros\IEEEauthorrefmark{1}}
		
		\IEEEauthorblockA{\IEEEauthorrefmark{1}Department of Electronic and Electrical Engineering, University College London, UK}
		
		Emails: \IEEEauthorrefmark{1}\{kaitao.meng, c.masouros\}@ucl.ac.uk
	}
	
	
	\maketitle

%%%%%%%%%%%%%%%%%%%%%%%%%%%%%%%%%%%%%%%%%%%%%%%%%%%%%%%%%%%%%%%%%%%%%%%%%%%%%%%%%

\begin{abstract}
In this work, we study integrated sensing and communication (ISAC) networks intending to effectively balance sensing and communication (S\&C) performance at the network level. 
Through the simultaneous utilization of multi-point (CoMP) coordinated joint transmission and distributed multiple-input multiple-output (MIMO) radar techniques, we propose a cooperative networked ISAC scheme to enhance both S\&C services. Then, the tool of stochastic geometry is exploited to capture the S\&C performance, which allows us to illuminate key cooperative dependencies in the ISAC network. Remarkably, the derived expression of the Cramer-Rao lower bound (CRLB) of the localization accuracy unveils a significant finding: Deploying $N$ ISAC transceivers yields an enhanced sensing performance across the entire network, in accordance with the $\ln^2N$ scaling law.
Simulation results demonstrate that compared to the time-sharing scheme, the proposed cooperative ISAC scheme can effectively improve the average data rate and reduce the CRLB.
\end{abstract}   

%\begin{IEEEkeywords}
%	Integrated sensing and communication, ISAC networks, network performance analysis, stochastic geometry, cooperative sensing, distributed radar. 
%\end{IEEEkeywords}
%%%%%%%%%%%%%%%%%%%%%%%%%%%%%%%%%%%%%%%%%%%%%%%%%%%%%%%%%%%%%%%%%%%%%%%%%%%%%%%%
\newtheorem{thm}{\bf Lemma}
\newtheorem{remark}{\bf Remark}
\newtheorem{Pro}{\bf Proposition}
\newtheorem{theorem}{\bf Theorem}
\newtheorem{Assum}{\bf Assumption}
\newtheorem{Cor}{\bf Corollary}

\section{Introduction}

The integration of sensing and communication (ISAC) emerges as a promising paradigm for next-generation networks \cite{LiuFan2022Integrated}. It employs unified spectrum, waveforms, platforms, and networks to address the issues of spectrum scarcity and circumvent interference caused by separate sensing and communication (S\&C) systems \cite{meng2024integrated}. ISAC can significantly enhance the spectrum, cost, and energy efficiency of S\&C functionalities \cite{Cui2021Integrating}. In the literature, most existing studies on this topic primarily concentrate on the ISAC design at the link/system level \cite{Meng2023Throughput}, but only limited works consider the ISAC design at the network level. 

The network-level ISAC is expected to provide several pronounced benefits compared to conventional single-cell ISAC. In terms of sensing, the ISAC network can expand its coverage to encompass larger surveillance areas, diverse sensing angles, and richer target information by forming multi-static sensing \cite{Shin2017CoordinatedBeamforming}. On the communications front, various ISAC transceivers can collaboratively utilize advanced coordinated multi-point (CoMP) transmission/reception techniques to enhance inter-cell interference management by connecting a single user with multiple base stations (BSs) \cite{Hosseini2016Stochastic}.
Despite the above advantages, networked ISAC brings new technical challenges for wireless resource allocation and user/target scheduling, as it requires a precise quantitative description of the average S\&C performance across the entire ISAC network. 

Stochastic geometry (SG) provides a powerful mathematical tool for communication network analysis \cite{Andrews2011TractableApproach}. For instance, \cite{Andrews2011TractableApproach} proposed a general framework for analysis of the average data rate and the coverage probability in multi-cell communication networks. Furthermore, \cite{chen2022isac} studied an ISAC beam alignment approach for THz networks, where SG is utilized to derive the coverage probability and network throughput performance. In a most recent work \cite{meng2023network}, the coordinated beamforming technique was implemented to mitigate interference for ISAC networks, offering insights into spatial resource allocation. However, the existing literature rarely explores the prospect of leveraging inter-cell interference to enhance the performance of ISAC networks.

Based on the above discussion, we propose a cooperative ISAC scheme that integrates CoMP joint transmission and multi-static sensing. By exploring the new degrees of freedom, i.e., optimizing cooperative BS cluster sizes for S\&C, we provide a promising solution to strike a balance between S\&C performance at the network level.
Then, we quantify ISAC performance through the data rate and Cramér-Rao lower bound (CRLB) \cite{Liu2022Integrated}, and apply SG techniques to conduct performance analysis. This analysis yields insights, directing emerging trends concerning both data rate and CRLB with increasing cooperative S\&C cluster size. 
The main contributions of this paper are summarized as follows:
\begin{itemize}[leftmargin=*]
	\item First, we propose a cooperative ISAC network framework, enabling the realization of CoMP joint transmission and distributed radar with the constraints of backhaul capacity. By incorporating the random BS locations, we derive the scaling law of the CRLB with respect to the BS number, i.e., $\ln^2 N$. We derive a tractable expression of the communication performance with flexible cooperative cluster size. 
	\item Second, in simulations, it is revealed that when provided with more resource blocks and larger backhaul capacity, the proposed cooperative scheme exhibits a greater performance improvement compared to the time-sharing scheme.
\end{itemize}

%Notation: $B(a,b,c) = \int_0^a t^(b-1) (1-t)^{c-1}dt$ is the incomplete Beta function. %Lowercase letters in bold font will denote deterministic vectors. For instance, $X$ and ${\bf{X}}$ denote one-dimensional (scalar) random variable and random vector (containing more than one element), respectively. Similarly, $x$ and ${\bf{x}}$ denote scalar and vector of deterministic values, respectively. ${\rm{E}}_{x}[\cdot]$ represents statistical expectation over the distribution of $x$, and $[\cdot]$ represents a variable set. %${\cal{O}}(0,r)$ denotes the circle region with center at the origin and radius $r$.

\section{System Model}

\begin{figure}[t]
	\centering
	\includegraphics[width=7.5cm]{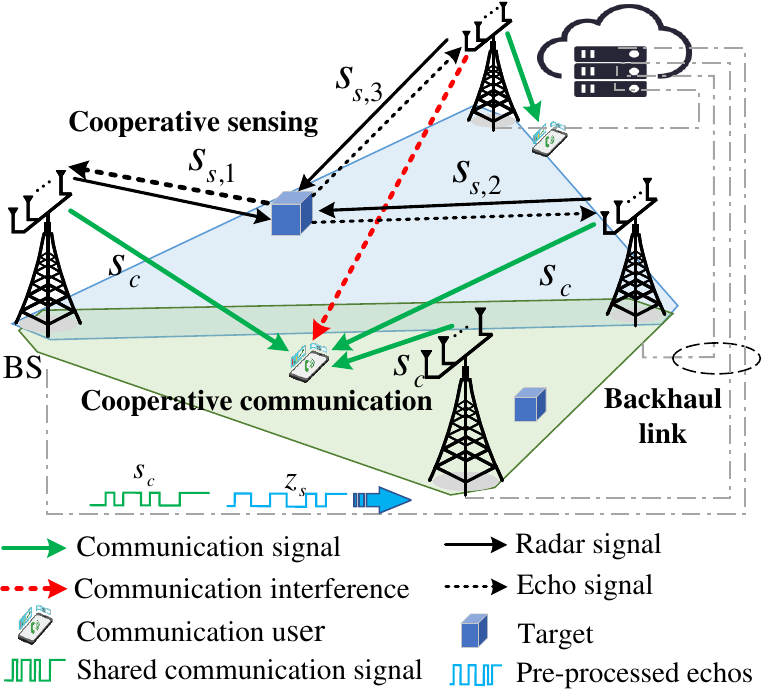}
	\vspace{-3mm}
	\caption{Illustration of cooperative S\&C networks.}
	\label{figure1}
\end{figure}
\subsection{Cooperative ISAC Network Model}
In the considered network, each BS is equipped with $M_{\mathrm{t}}$ transmit antennas and $M_{\mathrm{r}}$ received antennas, and the BS' location follows a homogeneous Poisson point process (PPP) on two-dimensional (2D) space, denoted by $\Phi_b$. Similarly, $\Phi_u$, and $\Phi_s$ respectively represent the point processes modeling the locations of communication users with a single antenna and targets. $\Phi_b, \Phi_u$, and $\Phi_s$ are mutually independent PPPs with intensities $\lambda_b$, $\lambda_u$, and $\lambda_s$, where $\Phi_b = \{ {\bf{d}}_i \in \mathbb{R}^2, \forall i \in \mathbb{N}^+\}$. 

Each communication user is served by $L \ge 1$ cooperative BSs that jointly transmit the same communication data, where the cooperative BS cluster is dynamically formed based on the user's location. Similarly, $N \ge 1$ BSs collaboratively provide localization service for each target, forming distributed multi-static multiple-input multiple-output (MIMO) radars, as shown in Fig.~\ref{figure1}. Each BS utilizes transmit beamforming to send information-bearing signal $s^c$ for the served user, together with a dedicated radar signal $s^{s}_i$ for the detected target. Following the assumption in \cite{Liu2020JointTransmit}, $E[s^s_i (s^c_i)^H] = 0$. Let ${\mathbf{s}_i=\left[s^s_i, s^c_i\right]^T}$, the transmitted signal at the $i$th BS is given by
\vspace{-1.5mm}
\begin{equation}\label{TrasmitSignals}
	{\mathbf{x}}_i = \mathbf{W}_i \mathbf{s}_i = \sqrt{p^c} {\bf{w}}^c_i s^c_i +  \sqrt{p^s} {\bf{w}}^s_i s^s_i,
	\vspace{-1.5mm}
\end{equation}
where ${\bf{w}}^c_i$ and ${\bf{w}}^s_i \in {\cal{C}}^{M_{\mathrm{t}} \times 1}$ are normalized transmit beamforming vectors, $p^s$ and $p^c$ represent the transmit power for S\&C, with $p^s + p^c = 1$ for normalized total transmit power, and $\mathbf{W}_i=\left[p^c \mathbf{w}^c_i, p^s \mathbf{w}^s_i\right] \in {\cal{C}}^{M_{\mathrm{t}} \times 2}$ is the precoding matrix of the BS at $\mathbf{d}_i$. To avoid interference between S\&C, we adopt zero-forcing beamforming to make the analysis tractable. Then, the beamforming vector of the serving BS $i$ can be given by 
\vspace{-1.5mm}
\begin{equation}\label{TransmitBeamforming}
	{\bf{W}}_i = {{\bf{H}}_i}{\left( {\bf{H}}_i ^H {\bf{H}}_i \right)^{-1}},
	\vspace{-1.5mm}
\end{equation}
where ${\bf{H}}_i = [({\bf{h}}^H_{i,c})^T, ({\bf{a}}^H(\theta_i))^T]^T$, ${\bf{h}}^H_{i,c}$ denotes the communication channel from BS $i$ to the typical user, and ${{\bf{a}}^H}(\theta_i )$ represents the sensing channel from BS $i$ to the typical target.

\subsection{Cooperative Sensing Model}
The location of the typical target is denoted by ${\psi}_t = [x_t, y_t]^T$. Assuming that unbiased measurements can be made, we have
\vspace{-1.5mm}
\begin{equation}	
	{\rm{var}}\{\hat{\psi}_t\} = {\rm{E}} [| \hat{\psi}_t - {\psi}_t |^2] \ge \mathrm{CRLB},
	\vspace{-1.5mm}
\end{equation}
where $\hat{\psi}_t=\left[\hat{x}_t, \hat{y}_t\right]^T$ represents the estimated location of the typical target. Assume that the transmitted radar signals $\{s_i^s\}_{i=1}^N$ of the BSs in the cooperative sensing cluster are approximately orthogonal for any time delay of interest \cite{li2008mimo}. Non-coherent MIMO radar is considered in this work due to its practical operational advantages \cite{sadeghi2021target}. Then, the base-band equivalent of the received signal at receiver $j$ is represented as 
\begin{equation}\label{SensingChannel}
	\begin{aligned}
		{{\bf{y}}_{j}}(t) =& \sum\nolimits_{i = 1}^N \sigma  {{{\left\| {{{\bf{d}}_j}} \right\|}^{ - \frac{\beta}{2} }}{\bf{b}}\left( {{\theta _j}} \right){{\left\| {{{\bf{d}}_i}} \right\|}^{ - \frac{\beta}{2} }}{{\bf{a}}^H}\left( {{\theta _i}} \right)}\sqrt{p^s}{{\bm{w}}^s_i}\\
		 &\times {s^s_i}\left( {t - {\tau _{i,j}}} \right) + {\bf{n}}_l(t),
	\end{aligned}
\end{equation}
where $\beta \ge 2$ is the pathloss exponent from the serving BSs to the typical target, $\sigma$ denotes the radar cross section (RCS), $\tau _{i,j}$ denotes the propagation delay of the link from BS $i$ to the typical target and then to BS $j$, and the term ${\bf{n}}_l(t)$ is the additive complex Gaussian noise with zero mean and covariance matrix ${\bm{\Sigma}} = \sigma_s^2 {\bf{I}}_{M_{\rm{r}}}$. In (\ref{SensingChannel}), ${{\bf{a}}^H}(\theta_n ) = [1, \cdots, e^{ {j \pi(M_{\mathrm{t}}-1)  \cos(\theta_n) }}]^T$, and ${\bf{b}}(\theta_j ) = [1, \cdots, e^{ {j \pi(M_{\mathrm{r}}-1)  \cos(\theta_j) }}]$. %${\bf{v}}^H_i$ is the receive beamforming vector, 

By measuring the range of each monostatic link and bi-static link reflected by the typical target, the target location can be estimated by methods such as the maximum likelihood estimation (MLE) \cite{li1993maximum}. Then, the FIM for estimating the parameter vector ${\psi}_t$ for the considered non-coherent MIMO radar is equal to \cite{sadeghi2021target}
\begin{equation}\label{FIMexpression}
{\bf{F}}_N = |\zeta |^2 \sum\nolimits_{i = 1}^N \sum\nolimits_{j = 1}^N {\left\| {{{\bf{d}}_i}} \right\|}^{ - \beta }{\left\| {{{\bf{d}}_j}} \right\|}^{ - \beta } {  {\left[ {\begin{array}{*{20}{c}}{a_{ij}^2}&{{a_{ij}}{b_{ij}}}\\
		{{a_{ij}}{b_{ij}}}&{b_{ij}^2}
\end{array}} \right]} } ,
\vspace{-1.5mm}
\end{equation}
where ${a_{ij}} = \cos {\theta _i} + \cos {\theta _j}$, ${b_{ij}} = \sin {\theta _i} + \sin {\theta _j}$, and the angle $\theta_i$ is the bearing angle of the $i$-th transmitter to the target with respect to the horizontal axis. In (\ref{FIMexpression}), it follows that $|\zeta |^2 = \frac{p^s G_t G_r B^2 \sigma}{8 \pi f_c^2 \sigma_s^2}$,
where $G_t$ and $G_r$ denote the transmit beamforming gain and receive beamforming gain, $f_c$ is the carrier frequency, and $B^2$ represents the squared effective bandwidth. With the random location of ISAC transceivers, the expected CRLB for any unbiased estimator of the target position is given by
\vspace{-1.5mm}
\begin{equation}	
	 \mathrm{CRLB}= {\rm{E}}_{\Phi_b} \left[\operatorname{tr}\left(\bar {\mathbf{F}}_N^{-1}\left(\psi_t\right)\right)\right].
	 \vspace{-1.5mm}
\end{equation}

\subsection{Cooperative Communication Model}

In this work, we consider the practical implementation of non-coherent joint transmission. The closest BS to the typical user sends collaboration service requests to the other $L-1$ BSs. The set of BSs receiving the service requirement is denoted by $\Phi_c$. Each BS decides whether to accept the request based on its load, where the BSs accepting the request are denoted by $\Phi_a$.  
Following Slivnyak’s theorem, the typical user is assumed to be located at the origin, and its performance is analyzed to generally represent the average performance of all users \cite{Andrews2011TractableApproach}. The index of the closest BS to the typical user is 1. Then, the received signal at the typical user can be given by 
\vspace{-1.5mm}
\begin{equation}
	\begin{aligned}
		y_{c}=& {\left\|\mathbf{d}_1\right\|^{-\frac{\alpha}{2}} \mathbf{h}_{1}^H \mathbf{W}_1 \mathbf{s}_1} + {\sum\nolimits_{i \in \Phi_a }\left\|\mathbf{d}_i\right\|^{-\frac{\alpha}{2}} \mathbf{h}_{i}^H \mathbf{W}_i \mathbf{s}_1}  \\
		&+{\sum\nolimits_{{j \in \{\Phi_b \backslash \Phi_a \backslash \{1\}\}}}\left\|\mathbf{d}_j\right\|^{-\frac{\alpha}{2}} \mathbf{h}_{j}^H \mathbf{W}_j \mathbf{s}_j} + {n_{c}},
		\vspace{-1.5mm}
	\end{aligned}
\end{equation}
where $\alpha \ge 2$ is the pathloss exponent, $\mathbf{h}^H_{i} \sim \mathcal{C N}\left(0, \mathbf{I}_{M_{\mathrm{t}}}\right)$ is the channel vector from the BS at $\mathbf{d}_i$ to the typical user, and $\Phi_a$ is the cooperative BS set. This paper focuses on evaluating the performance of an interference-limited network within dense cell scenarios. Thus, the evaluation is based on the signal-to-interference ratio (SIR) \cite{Park2016OptimalFeedback}.
The SIR of the received signal at the typical user can be expressed as
\vspace{-1.5mm}
\begin{equation}\label{SIRexpression}
	{\rm{SI}}{{\rm{R}}_c} = \frac{ {{g_1}{{\left\| {{{\bf{d}}_1}} \right\|}^{ - \alpha }}} +{\sum\limits_{i \in \Phi_a} {{g_i}{{\left\| {{{\bf{d}}_i}} \right\|}^{ - \alpha }}} }}{{\sum\limits_{j \in \{\Phi_b \backslash \Phi_a \backslash \{1\}\}}  {{g_j}} {{\left\| {{{\bf{d}}_j}} \right\|}^{ - \alpha }}}},
	\vspace{-1.5mm}
\end{equation}
where $g_{1}= p^c\left|\mathbf{h}_{1}^H \mathbf{w}_1^c\right|^2$ and $g_{i}=p^c\left|\mathbf{h}_{i}^H \mathbf{w}_i^c\right|^2$ denotes the effective desired signals' channel gain, and the interference channel gain $g_{j} = p^c\left|\mathbf{h}_{j}^H \mathbf{w}_j^c\right|^2 + p^s\left|\mathbf{h}_{j}^H \mathbf{w}_j^s\right|^2$. The average data rate of users can be given by 
\vspace{-1.5mm}
\begin{equation}
	R_c=\mathrm{E}_{\Phi_b,g_i}[\log (1+\mathrm{SIR}_c)].
	\vspace{-1.5mm}
\end{equation}

\subsection{Limited Backhaul Capacity Model}
In the considered system, each BS is connected to the central unit through backhaul links to share the data and CSI information for CoMP transmission and to collect the echo signals for cooperative sensing, as shown in Fig.~\ref{figure1}. 
It is noteworthy that in a cooperative S\&C system, the sizes of cooperative clusters are practically restricted by the constrained capacity of the backhaul link \cite{Ghimire2015Revisiting}. Therefore, backhaul constraints introduce an additional dimension to strike the balance between S\&C performance. For clarity, we exclude consideration of backhaul traffic related to CSI sharing \cite{Zhang2013Downlink}, and the capacity constraint can be expressed as
\vspace{-1.5mm}
\begin{equation}
	R_c + e \times N \le C_{\text{backhaul}},
	\vspace{-1.5mm}
\end{equation}
where $e$ represents the data rate demand for the sensing cooperation, i.e. for sending pre-processed results (i.e., auto-correlation of signal $\{s^s_i\}_{i=1}^N$ transmitted by each BS), and $C_{\text{backhaul}}$ denotes the backhaul capacity limitation.

\section{Sensing Performance Analysis}
\label{SensingSection}

\subsection{Performance Gain of Cooperative Sensing}
First, by assuming all BSs accept the requests, we derive the closed-form CRLB expression with the consideration of random locations of BSs and targets, based on which, the scaling law of localization accuracy is obtained. First, the CRLB expression can be equivalently transformed into
\vspace{-1.5mm}
\begin{equation}
	\begin{aligned}
		&{\rm{CRLB}} = {{\rm{E}}_{\Phi_b}}\bigg[|\zeta |^{-2} \times \\
		& \frac{{2\sum\nolimits_{i = 1}^N \!{\sum\nolimits_{j = 1}^N d_i^{-\beta}d_j^{ - \beta}\left(1+\cos \left( {{\theta _i} - {\theta _j}} \right)\right) } } } {{\!\!\sum\nolimits_{l = 1}^N {\!\! \sum\nolimits_{k = 1}^N {\! \sum\nolimits_{i \ge k}^N {\! \sum\nolimits^N_{j > \!{\lceil (k - i)N + l \rceil}^+}\! \!{(d_i d_j d_l d_k)^{\!-\beta}} } } } {{\! \left( {{a_{kl}}{b_{ij}} \!-\! {a_{ij}}{b_{kl}}} \right)}^2}}}\! \bigg] \!,
		\vspace{-1.5mm}
	\end{aligned}
\end{equation}
where $d_i = {\left\| {{{\bf{d}}_i}} \right\|}$ and ${\lceil x \rceil}^+ = \max(x,1)$.
To obtain a more tractable CRLB expression, we resort to a simple but tight approximation, and then the following conclusion is proved.
\begin{Pro}\label{SimplifiedWithDis1}
The expected CRLB can be approximated as 
\vspace{-1.5mm}
\begin{equation}\label{SimplifiedExpressionCRLB}
	{\rm{CRLB}} =	\frac{2}{|\zeta |^{2}{\sum\nolimits_{l = 1}^N {\sum\nolimits_{k = 1}^N {E{{\left[ {{d_k}} \right]}^{ - \beta}}E{{\left[ {{d_l}} \right]}^{ - \beta}}} } }}.
	\vspace{-1.5mm}
\end{equation}
\end{Pro}
\begin{proof}
	Please refer to Appendix A.
\end{proof}

Interestingly, it can be found that the value of expected CRLB in Proposition \ref{SimplifiedWithDis1} is only determined by the expected distance from the BS to the typical target. It is verified that (\ref{SimplifiedExpressionCRLB}) achieves a good approximation by Monte Carlo simulations, as shown in Section \ref{simulations}.
Furthermore, the expected distance from the $n$th closest BS to the typical target can be expressed as $E \left[ {{d_n}} \right] =  { {\frac{{\Gamma \left( n + \frac{1}{2} \right)}}{{\sqrt{\lambda_b \pi} \Gamma (n)}}} } \approx \sqrt{\frac{n}{\lambda_b \pi}}$. Then, the CRLB expression can be further approximated as
\vspace{-1.5mm}
\begin{equation}\label{CRLB_expression}
	{\rm{CRLB}} \approx \frac{2}{|\zeta |^{2}{{\lambda_b ^\beta}{\pi ^\beta}\sum\nolimits_{l = 1}^N {\sum\nolimits_{k = 1}^N { k^{-\frac{\beta}{2}} l^{-\frac{\beta}{2}}} } }}.
	\vspace{-1.5mm}
\end{equation}
With general setup $\beta = 2$, we further derive the scaling law of the localization accuracy as follows.

\begin{theorem}\label{SimplifiedWithDis3}
	With infinity cooperative cluster size $N$, the expected CRLB can be given by
	\vspace{-1.5mm}
	\begin{equation}
		\mathop {\lim }\limits_{N \to \infty } {\rm{CRLB}} \times {{\ln }^2}N = \frac{1}{|\zeta |^{2}{{\lambda_b ^2}{\pi ^2}}}.
		\vspace{-1.5mm}
	\end{equation}
\end{theorem}
\begin{proof}
	Please refer to Appendix B.
\end{proof}

The CRLB scaling law ${{\ln }^2}N$ derived in Theorem \ref{SimplifiedWithDis3} is highly useful to indicate the cooperative sensing design. 

\subsection{Acceptance Probability with Limited Resource Blocks}
\label{AccperationProbability}
Using the target-centric clustering model, the center unit sends localization service requests to the $N$ closest BSs relative to the typical target. Let $\psi$ be an integer representing the maximum load, i.e., the maximum number of targets that can be simultaneously served by the BS.
Then, if a BS receives $N$ requests, we assume it will randomly choose $\psi$ targets to provide services. In this case, the acceptance probability of the BS which receives service requests can be given as follows:
\begin{thm}\label{AcceptationProbability}
	When each target requests $N$ BSs to provide localization services, the acceptance probability of the BS can be given by:
	\vspace{-1.5mm}
	\begin{equation}
		\kappa_s = \frac{\Gamma {\left( \psi ,\mu_s \bar N \right)}}{{(\psi  - 1)!}} + \sum\limits_{n = \psi  + 1}^\infty  {\frac{{ {\psi} {{\left( {\mu_s \bar N} \right)}^n}}}{{n \times n!}}} {e^{ - \mu_s \bar N}},
		\vspace{-1.5mm}
	\end{equation}
	where $\mu_s = {\frac{{{\lambda _s}}}{{{\lambda _b}}}}$ and $\bar N = {\frac{\Gamma(N+\frac{1}{2})^2}{\Gamma(N)^2}}$.
\end{thm}
\begin{proof}
	Please refer to Appendix C.
\end{proof}

According to Lemma \ref{AcceptationProbability}, the acceptance probability is monotonically increasing with the number of resource blocks $\psi$, and is monotonically decreasing with the target-BS density ratio $\mu_s$ and cluster size $N$. 
Then, the expected CRLB with the consideration of acceptance probability can be expressed as 
\vspace{-1.5mm}
\begin{equation}
	{\rm{CRLB}}_a =  \frac{1}{\kappa_s ^2 |\zeta |^{2}{{ \lambda_b ^2}{\pi ^2}{{\ln }^2}N}}.
	\vspace{-1mm}
\end{equation}

\section{Communication Performance Analysis}
\label{CommunicationPerformance}
To implement CoMP joint transmission, the closest BS sends service requests to the other $L-1$ closest BSs to the typical user. Similarly, if a BS receives more than $\psi$ requests, it will randomly choose $\psi-1$ users to provide services besides the typical user. Similar to the sensing acceptance probability, when each user requests $L$ BSs to provide communication services, the acceptance probability of the BS can be given by $\kappa_c = \frac{{\Gamma \left( {\psi ,\mu_c \bar L} \right)}}{{(\psi  - 1)!}} + \sum\limits_{n = \psi}^\infty  {\frac{{\left( {\psi  - 1} \right){{\left( {\mu_c \bar L} \right)}^n}}}{{(n-1) \times n!}}} {e^{ - \mu_c \bar L}}$, where $\mu_c = \frac{\lambda_u}{\lambda_b}$ and $\bar L = {\frac{\Gamma(L+\frac{1}{2})^2}{\Gamma(L)^2}}$.

According to \cite{hamdi2010useful}, for uncorrelated variables $X$ and $Y$, we have
\vspace{-1mm}
\begin{equation}\label{CommunicationBasicEquation}
	{\rm{E}}\left[ {\log \left( {1 + \frac{X}{Y}} \right)} \right] \! = \! \int_0^\infty  {\frac{1}{z}} \left( {1 - {\rm{E}}\left[{e^{ - z \left[ X\right] }}\right]} \right){\rm{E}}\left[{e^{ - z\left[ Y \right]}}\right]{\rm{d}}z.
	\vspace{-1mm}
\end{equation}
Then, under a given distance $r$ from the typical user to the closest BS, the conditional expectation of data rate can be derived as follows:
\vspace{-1mm}
\begin{equation}
	\begin{aligned}
		&{\rm{E}}\left[ {\log \left( {1 + \mathrm{SIR}_c} \right)} \big| r \right] \\
		=& {\rm{E}} \! \left[ {\log \left( \! {1 + \frac{{g_{1}} + \sum\nolimits_{{{i}} \in {\Phi_a}} {g_{i}} \left\| {\bf{d}}_i \right\|^{-\alpha} r^{\alpha}}{ \sum\nolimits_{{{j}} \in \{\Phi_b \backslash \Phi_a \backslash \{1\}\}} g_j \left\| {\bf{d}}_j \right\|^{-\alpha} r^\alpha }} \right)} \right] \\
		=& \int_0^\infty  {\frac{{1 - {\rm{E}}\left[ {{e^{ - zg_{1}}}} \right] {\rm{E}}\left[ {{e^{ - z U}}} \right]}}{z}} {\rm{E}}\left[ {{e^{ - z I_{1}}}} \right] {\rm{E}}\left[ {{e^{ - z I_{2}}}} \right]{\rm{d}}z,
		\vspace{-1mm}
	\end{aligned}
\end{equation}
where $U = \sum\nolimits_{{{i}} \in {\Phi_a}} {g_{i}} r^{\alpha}$, $ I_{1} = \sum\nolimits_{{{i}} \in \{\Phi_c \backslash \Phi_a\}}  {{g_{i}}} {{\left\| {{{\bf{d}}_{{i}}}} \right\|}^{ - \alpha }}{r^\alpha }$, and $ I_{2} = \sum\nolimits_{{{i}} = L+1}^\infty  {{g_{i}}} {{\left\| {{{\bf{d}}_{{i}}}} \right\|}^{ - \alpha }}{r^\alpha }$. Here, $I_1$ represents the interference from the BSs declining the cooperation requests, and $I_2$ represents the interference from the BS located beyond the cooperative request cluster. $g_{1}$ and $g_{i}$ are the effective desired signal channel gain, $g_{1}, g_{i} \sim \Gamma \left( M_{\mathrm{t}} - 1, p^c\right)$ \cite{Hosseini2016Stochastic}. According to the definition below (\ref{SIRexpression}), we can derive the distribution of $g_{j}$ based on the moment matching technique \cite{Hosseini2016Stochastic}, and it follows that $g_{j} \sim \Gamma (1 , 1)$. Thus, the useful signal power can be given by ${\rm{E}}\left[ {{e^{ - zg_{1}}}} \right] = {\left( {1 + p^cz} \right)^{1 - {M_{\mathrm{t}} }}}$.
Then, we derive tight bounds on the Laplace transform of cooperative transmission power and communication interference as follows:

\begin{thm}\label{LaplaceTransform}
	With the closest BS at a distance $r$, the Laplace transforms of $U$, $I_1$, and $I_2$ can be given by
	\vspace{-1mm}
	\begin{equation}
		{\rm{E}}\left[ {{e^{ - z U}}} \right] \!= \exp (  - \pi \kappa_c \lambda_b {r^2}{\rm{H}}_1( { zp^c,M_{\mathrm{t}}-1,\alpha ,\eta_{L} } ) ),
	\end{equation}
	\begin{equation}
		{\rm{E}}\left[ {{e^{ - zI_1}}} \right] =   \exp (  - \pi (1-\kappa_c ) \lambda_b {r^2}{\rm{H}}_1( {z,1,\alpha ,\eta_{L} } )  ),
	\end{equation}
	\begin{equation}
		{\rm{E}}\left[ {{e^{ - zI_2}}} \right] = \exp (  - \pi \lambda_b {r^2}{\rm{H}}_2( {z,\alpha ,\eta_{L} } ) ),
	\end{equation}
where ${\rm{H}}_1\left( {x,K,\alpha ,\eta_L } \right)  = \frac{1}{{{\eta ^2}}}\left( {1 - \frac{1}{{{{\left( {1 + x{\eta ^\alpha }} \right)}^K}}}} \right) + \frac{1}{{{{\left( {1 + x} \right)}^K}}} - 1 + K{x^{\frac{2}{\alpha }}}\!\left(\! {B\left( {\frac{x}{{x + 1}},1 \! - \! \frac{2}{\alpha },K + \frac{2}{\alpha }} \right) \! - \! B\left( {\frac{{x{\eta ^\alpha }}}{{x{\eta ^\alpha } + 1}},1 - \frac{2}{\alpha }, K + \frac{2}{\alpha }} \right)} \right)$, ${\rm{H}}_2\left( {x,\alpha ,\eta_L } \right) = {x^{\frac{2}{\alpha }}}B\left( {\frac{x}{{x + {\eta_L ^{ - \alpha }}}},1 - \frac{2}{\alpha },1 + \frac{2}{\alpha }} \right) + \frac{1}{{{\eta_L ^2}}}\left( {{{{{\left( {1 + x{\eta_L ^\alpha }} \right)}^{-1}}}} - 1} \right)$, $\eta_{L} = \frac{r}{r_L}$, and $r_L = \|{\bf{d}}_L \|$. Here, $B(a,b,c) = \int_0^a t^(b-1) (1-t)^{c-1}{\rm{d}}t$ is the incomplete Beta function.
\end{thm}
\begin{proof}
	Please refer to Appendix D.
\end{proof}

Based on the obtained Laplace transforms of $U$, $I_1$, and $I_2$, the expected data rate can be obtained in Theorem \ref{CommunicationTightExpression}.

\begin{theorem}\label{CommunicationTightExpression}
	The communication performance can be given by $R_c = \int_0^\infty  \int_0^1 \frac{2\left( {L - 1} \right)\eta_L {\left( {1 - {\eta_L ^2}} \right)^{L - 2}}}{z}\bigg( \frac{1}{{\left( {1 - \kappa_c } \right){{\rm{H}}_1}\left( {z,1,\alpha,\eta_L } \right) + {{\rm{H}}_2}\left( {z,\alpha,\eta_L } \right) + 1}} -  
	\frac{{{{\left( {1 + p^c z} \right)}^{ 1- {M }_{\mathrm{t}} }}}}{{\kappa_c {{\rm{H}}_1}\left( {z p^c,M_{\mathrm{t}} - 1,\alpha,\eta_L } \right) + \left( {1 - \kappa_c } \right){{\rm{H}}_1}\left( {z,1,\alpha,\eta_L } \right) + {{\rm{H}}_2}\left( {z,\alpha,\eta_L } \right) + 1}} \bigg){\rm{d}}\eta_L {\rm{d}}z$
\end{theorem}
\begin{proof}
	Please refer to Appendix E. 
\end{proof}

According to Theorem \ref{CommunicationTightExpression}, the average data rate increases with the BS density $\lambda_b$ and the number of resource blocks.

\section{Simulation Results}
\label{simulations}
The numerical simulations are averaged over various network typologies and realizations of small-scale channel fading. The system parameters are given as follows: the number of transmit antennas $M_{\mathrm{t}} = 4$, the number of receive antennas $M_{\mathrm{r}} = 5$, the transmit power $P_{\mathrm{t}} = 1$W at each BS, the average RCS $\sigma = 1$, the BS density $\lambda_b = 1/km^2$, $\lambda_u = 1/km^2$, $\lambda_s = 1/km^2$, $\sigma^2_s = - 80$dB, pathloss coefficients $\alpha = 4$, $\beta = 2$, backhaul capacity $C_{\text{backhaul}} = 6$, and resource block number $\psi = 15$.

\subsection{Sensing Performance}

\begin{figure*}[!t]
	\begin{minipage}[t]{0.33\linewidth}
		\centering
		\includegraphics[width=5.2cm]{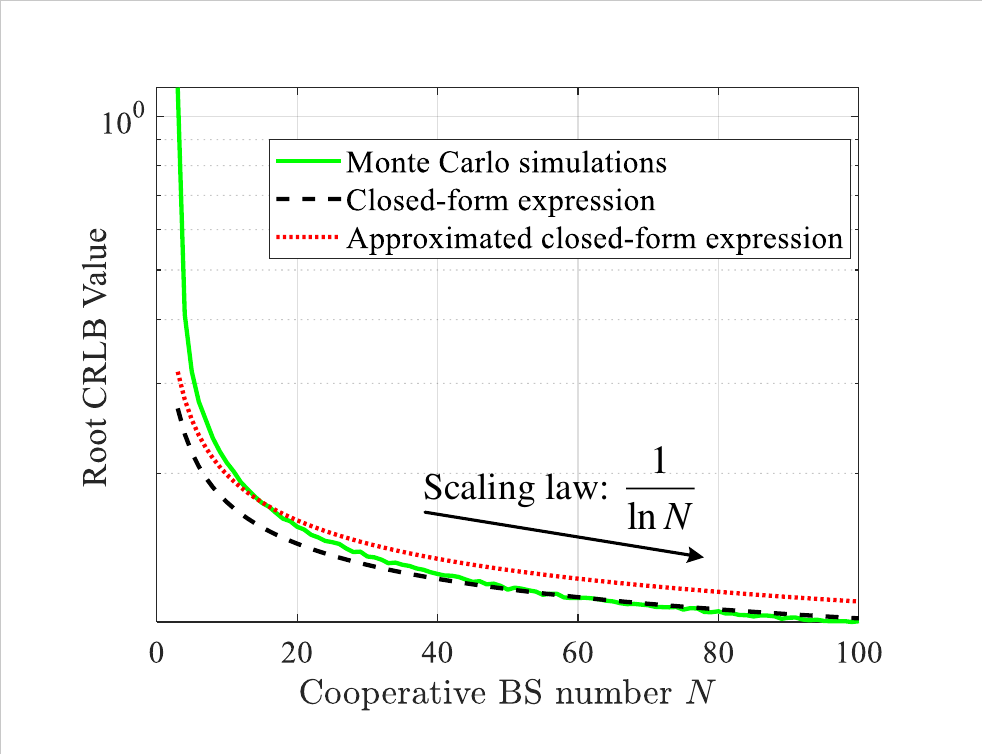}
		\vspace{-3mm}
		\caption{CRLB value comparison with $\kappa_c = 1$.}
		\label{figure6}
	\end{minipage}%
	\begin{minipage}[t]{0.33\linewidth}
		\centering
		\includegraphics[width=5.4cm]{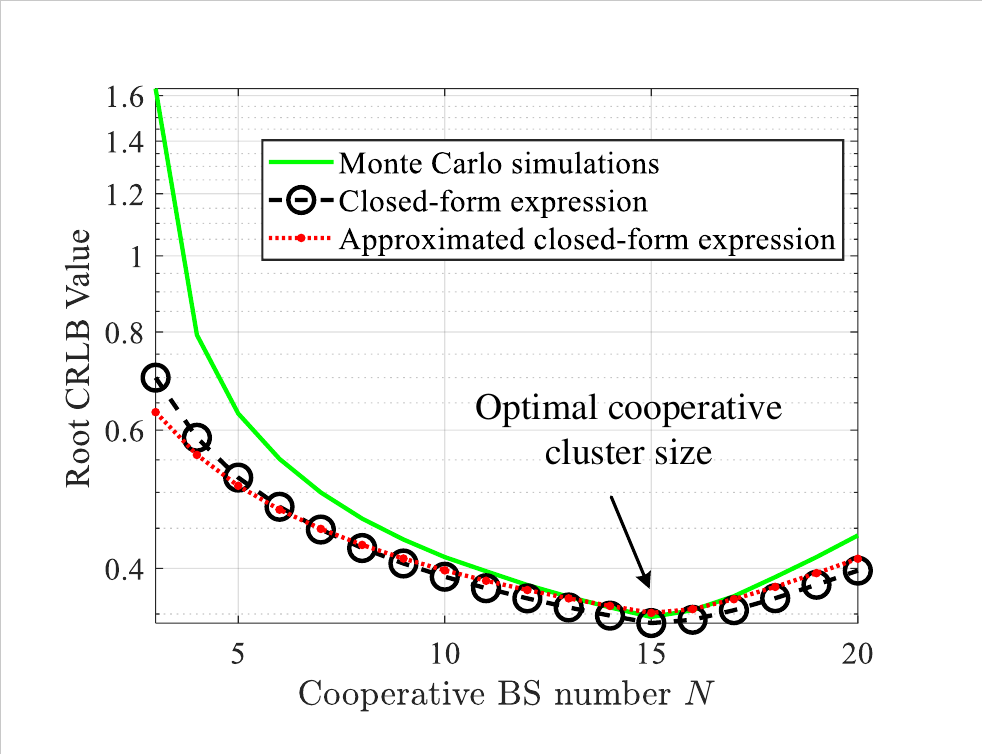}
		\vspace{-3mm}
		\caption{CRLB value with the acceptance probability.}
		\label{figure4}
	\end{minipage}
	\begin{minipage}[t]{0.33\linewidth}
		\centering
		\includegraphics[width=5.2cm]{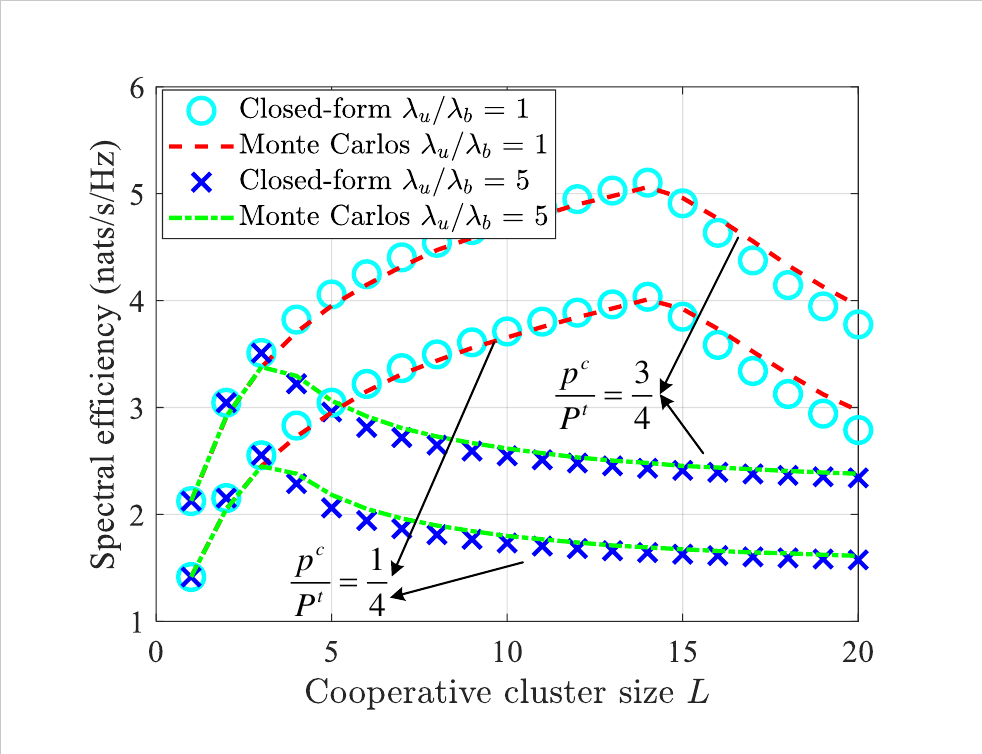}
		\vspace{-3mm}
		\caption{Cooperative communication performance.}
		\label{figure7}
	\end{minipage}
\end{figure*}

In Fig.~\ref{figure6}, the tractable expression derived in Theorem \ref{SimplifiedWithDis3} provides a remarkably tight approximation, especially for larger cooperative BS number $N$. It is noteworthy that when the number of cooperative BSs is relatively small, for instance, $N \le 4$, the closed-form expressions exhibit a slight deviation from Monte Carlo simulations. This is mainly due to the less precise calculation of the expectation operation involving trigonometric functions when the number of ISAC BSs is small.
Furthermore, Fig.~\ref{figure6} reveals that augmenting the number of cooperative BSs results in substantial accuracy improvement when the total number of BSs is limited, yet it yields only incremental performance gains for $N \ge 10$. This is expected because more participation of randomly located BSs in cooperation leads to increased signal attenuation for distant BSs, resulting in a performance gain that is significantly less than that observed for nearby BSs. 

Furthermore, with $\psi = 15$, Fig.~\ref{figure4} shows that the CRLB decreases first and then increases when $N \ge 15$, i.e., the optimal size of the cooperative sensing cluster equals $\psi$. The main factor is that, as the average number of service targets per BS exceeds the allocated resource blocks, each BS is likely to reach full load. Consequently, some requests sent from nearby targets may be declined, resulting in the forming of a cooperative sensing cluster with BSs situated at a larger distance from the typical target. As the average number of service requests continues to rise, the distance between transceivers and targets also increases.

Fig.~\ref{figure7} illustrates that the results of the original expression for $R_c$ in Theorem \ref{CommunicationTightExpression} are consistent with the simulation results. 
As depicted in Fig.~\ref{figure7}, for any given communication power ${p^c}$, the communication spectral efficiency $R_c$ initially increases and then decreases with the cooperative cluster size $L$. The primary cause for this trend is the increasing involvement of more users in the service, leading to a reduction in the average acceptance probability for each user. It is evident that, across various power ${p^c}$, there is a consistent optimal value $L^*$ for the same user-BS density $\frac{\lambda_u}{\lambda_b}$ that maximizes spectral efficiency. Meanwhile, this optimal value $L^*$ decreases as $\frac{\lambda_u}{\lambda_b}$ increases. This is attributed to the fact that with a higher $\frac{\lambda_u}{\lambda_b}$, more users may send service requests to the BSs, consequently increasing the load of each BS and diminishing the acceptance probability.

\begin{figure}[t]
	\centering
	\includegraphics[width=6cm]{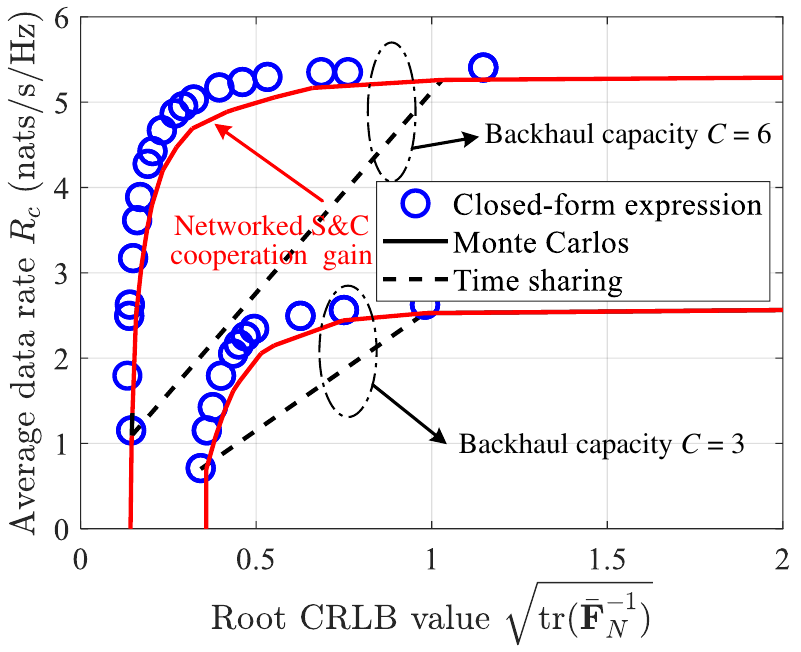}
	\vspace{-3mm}
	\caption{Performance tradeoff versus different backhaul capacity constraints.}
	\label{figure9}
\end{figure}

We validate the proposed cooperative ISAC scheme, encompassing both the S\&C performance boundary. First, we compare the effectiveness of the time-sharing scheme based on two corner points to illustrate the performance of the cooperative ISAC scheme under various setups.
The tradeoff profile between the average data rate $R_c$ and the average ${\rm{CRLB}}_a$ is depicted in Fig. \ref{figure9}, confirming both the accuracy of analytical results and the flexibility of our proposed cooperative ISAC networks. As the backhaul capacity increases, the performance boundaries of S\&C expand significantly. Also, it is observed from Fig. \ref{figure9} that the achievable S\&C performance region of the optimal cooperative scheme becomes much larger than that of the time-sharing scheme as the backhaul capacity increases. 
 This outcome stems from the augmented capacity of backhaul links, enabling the network to effectively coordinate transmit power and multi-cell resources, thereby enhancing the cooperative cluster design gains for S\&C.

%Fig.~\ref{figure10} shows that the networked ISAC performance can be effectively extended by exploiting the optimal cooperative strategy under different user-BS densities. Specifically, the proposed cooperative scheme can improve the communication performance by up to 78\% and 50\% as compared to the time-sharing scheme with $\frac{\lambda_u}{\lambda_b} = 1$ and $\frac{\lambda_u}{\lambda_b} = 5$, respectively. Similar to the average S\&C performance depicted in Fig.~\ref{figure9}, the rate-CRLB region for the proposed cooperative ISAC scheme significantly expands compared to the time-sharing scheme as backhaul capacity grows. In addition, as shown in Fig. \ref{figure10}, as the sensing performance decreases under small root CRLB value, the communication rate improves more significantly. The main reason is that when the number of BSs is small, the sensing performance improves rapidly as the number of BSs increases. 

\section{Conclusion}
In this paper, we proposed a novel cooperative scheme in ISAC networks by simultaneously adopting the CoMP joint transmission and distributed radar techniques. With SG tools, the S\&C performance expressions are described analytically. We revealed that the average cooperative sensing performance CRLB in the entire ISAC network scales with $\ln^2N$. The simulation results demonstrate the benefits of the proposed cooperative ISAC scheme and provide insightful guidelines for designing practical large-scale ISAC networks.

\section*{Appendix A: \textsc{Proof of Proposition \ref{SimplifiedWithDis1}}}

First, the expected transmit beamforming gain can be given by ${\rm{E}}[\left|\mathbf{a}^H(\theta_i) \mathbf{w}_i^s\right|^2] = M_{\rm{t}} - 1$, i.e., $G_t = M_{\rm{t}} - 1$.
Since BS' location follows a homogeneous PPP, the angle $\theta_i$ and distance $d_i$ are independent for each BS. Thus, it follows that ${{\rm{E}}_{d,\theta }}\left[ {{\rm{tr}}\left( {{{\tilde {\bf{F}}}^{ - 1}}} \right)} \right] = {{\rm{E}}_d}\left[ {{{\rm{E}}_\theta }\left[ {{\rm{tr}}\left( {{{\tilde {\bf{F}}}^{ - 1}}} \right)} \right]} \right]$. Then, the CRLB can be transformed as shown in (\ref{SimplifiedDerivation}), at the top of the next page.

\begin{equation}\label{SimplifiedDerivation}
	\begin{aligned}
		{\rm{CRLB}} 
		& \! = \! {{\rm{E}}_d} \!\left[\! {\frac{{2{{\rm{E}}_\theta }[ {\sum\nolimits_{i = 1}^N {\sum\nolimits_{j = 1}^N {{D _{ij}}} (1+\cos ( {{\theta _i} - {\theta _j}} )) } } ]}}{{{{\rm{E}}_\theta }[ {\sum\nolimits_{l = 1}^N {\sum\nolimits_{k = 1}^N {\sum\nolimits_{i \ge k}^N {\sum\nolimits_{j > \!{\lceil (k - i)N + l \rceil}^+}^N \!\! {{D _{kl}}{D _{ij}}} } } } {Y}} ]}}} \!\right] \\
		&\approx {{\rm{E}}_d}\left[ {\frac{2}{{\sum\nolimits_{l \ne k}^N {\sum\nolimits_{k = 1}^N {{D _{kl}}} } }}} \right] \\ &\approx \frac{2}{{\sum\nolimits_{l \ne k}^N {\sum\nolimits_{k = 1}^N {E{{\left[ {{d_k}} \right]}^{ - \beta}}E{{\left[ {{d_l}} \right]}^{ - \beta}}} } }}.
		\vspace{-3mm}
	\end{aligned}
\end{equation}
In (\ref{SimplifiedDerivation}), $D_{i,j} = d_i^{-\beta}d_j^{-\beta}$, ${Y}^2 = \left( {{a_{kl}}{b_{ij}} - {a_{ij}}{b_{kl}}} \right)^2$ and the first approximation is adopted by ignoring the items with lower order, and the second approximation holds due to the independent distance of different BSs.

\section*{Appendix B: \textsc{Proof of Theorem \ref{SimplifiedWithDis3}}}
First, when $\beta = 2$, the denominator of (\ref{CRLB_expression}) can be expressed as $\sum\nolimits_{l = 1}^N {\frac{1}{l}\sum\nolimits_{k = 1}^N {\frac{1}{k}} } $.
According to the sum of series, we have ${H_N} = 1 + \frac{1}{2} + \frac{1}{3} +  \ldots . + \frac{1}{N} \approx \ln (N) + \gamma  + \frac{1}{{2N}}$ and $\mathop {\lim }\limits_{N \to \infty } \left[ {{H_N} - \ln (N)} \right] = \gamma ,\gamma  = 0.577$. When $N \to \infty$, it follows that
\vspace{-1mm}
	\begin{align}
		&\mathop {\lim }\limits_{N \to \infty } \sum\nolimits_{l = 1}^N {\frac{1}{l}\sum\nolimits_{k = 1}^N {\frac{1}{k}} } \!=\! \mathop {\lim }\limits_{N \to \infty } \!\sum\nolimits_{l = 1}^N {\frac{1}{l}\left( {\ln N\! + \!\gamma  + \frac{1}{{2N}}} \right)} \nonumber \\
		&= \mathop {\lim }\limits_{N \to \infty } {\left( {\ln \left( N \right) + \gamma  + \frac{1}{{2N}}} \right)^2} = {\ln ^2} N .
		\vspace{-1mm}
	\end{align}
This thus completes the proof.

\section*{Appendix C: \textsc{Proof of Lemma \ref{AcceptationProbability}}}
Under the proposed clustering model, the average association area of each BS is given by ${\cal{A}} = \frac{{{\Gamma ^2}\left( {N + 0.5} \right)}}{{{\lambda _b}{\Gamma ^2}\left( N \right)}}$.
Further, the users are PPP distributed with density $\lambda_s$, thus the average number of users associated with a BS is ${\cal{A}} \lambda_s$. The probability that a BS accepts the localization serving request is equivalent to the probability that the number of users within the BS association area is less than the maximum load, $\psi$. Let ${\cal{N}}\left(|{\cal{A}}| \lambda_u \right)$ be the number of users with density $\lambda_u$ in a geographical area $|{\cal{A}}|$. Let $\bar N = \frac{{{\Gamma ^2}\left( {N + 0.5} \right)}}{{{\Gamma ^2}\left( N \right)}}$, $\kappa_s$ can be calculated as follows:
\vspace{-1mm}
\begin{equation}
	\begin{aligned}
		\kappa_s  = & \sum\limits_{n = 0}^\psi  {\Pr } \left[ {{\cal N}\left( {{\lambda _u} |{{\cal{A}}} | = n} \right)} \right] + \sum\limits_{n = \psi  + 1}^\infty  {\Pr } \left[ {{\cal N}\left( {{\lambda _u} |{{\cal{A}}} | = n} \right)} \right]\frac{\psi }{n} \\
		=& \frac{{\Gamma \left( {\psi ,\mu_s \bar N} \right)}}{{\psi !}} + \sum\limits_{n = \psi  + 1}^\infty  {\frac{{\psi {{\left( {\mu_s \bar N} \right)}^n}}}{{n \times n!}}} {e^{ - \mu_s \bar N}}.
		\vspace{-1mm}
	\end{aligned}
\end{equation}
This thus completes the proof.

\section*{Appendix D: \textsc{Proof of Lemma \ref{LaplaceTransform}}}
The interference term with a given distance $r$ from the typical user to the closest BS, can be derived by utilizing Laplace transform. For ease of analysis, we introduce a geometric parameter $\eta_L = \frac{\left\| {{{\bf{d}}_1}} \right\|}{\left\| {{{\bf{d}}_{L}}} \right\|} $, defined as the distance to the closest BS normalized by the distance to the furthest BS in the cluster for typical user. When $\left\| {{{\bf{d}}_1}} \right\| = r$ and $\left\| {{{\bf{d}}_L}} \right\| = r_L$, we have
\vspace{-1mm}
\begin{align}\label{CommunicationEquationExpression}
	&{{\cal L}_{{I_{2}}}}(z)  =  {\rm{E}}_{\Phi_b, g_i} \!\left[ {\exp \left( { - z{{r}^\alpha }\sum\nolimits_{i = L+1}^\infty  \! {{{\left\| {{{\bf{d}}_i}} \right\|}^{ - \alpha }}} {{| {{\bf{h}}_{i}^H{{\bf{W}}_i}} |}^2}} \right)} \! \right] \nonumber  \\
	&\overset{(a)}{=}  {\rm{E}}_{\Phi_b}\!\left[ \!\left( \prod _{{{{\bf{d}}_i}} \in \Phi_b \textbackslash {\cal{O}}(0,r_L)} {   {  {{{{\left( {1 + z{r^\alpha }{{\left\| {{{\bf{d}}_i}} \right\|}^{ - \alpha }}} \right)}^{-1}}}} } dx} \right) \bigg| r, r_L \right]  \nonumber \\
	&\overset{(b)}{=} \exp \left( { - 2\pi \lambda_b \int_{{r_L}}^\infty  {\left( {1 - {{{{\left( {1 + z{r^\alpha }{x^{ - \alpha }}} \right)}^{-1}}}}} \right)} xdx} \right)  \nonumber \nonumber \\
	&\overset{(c)}{=} \exp \bigg(  - \pi \lambda_b {r^2}{\rm{H}}_2\left( {z,\alpha ,\eta_{L} } \right) \bigg) ,
	\vspace{-1mm}
\end{align}
where ${\rm{H}}_2\left( {x,\alpha ,\eta_L } \right) = {x^{\frac{2}{\alpha }}}B\left( {\frac{x}{{x + {\eta_L ^{ - \alpha }}}},1 - \frac{2}{\alpha },1 + \frac{2}{\alpha }} \right) + \frac{1}{{{\eta_L ^2}}}\left( {{{{{\left( {1 + x{\eta_L ^\alpha }} \right)}^{-1}}}} - 1} \right)$.
In (\ref{CommunicationEquationExpression}), ($a$) follows from the fact that the small-scale channel fading is independent of the BS locations and that the interference power imposed by each interfering BS at the typical user is distributed as $\Gamma(1,1)$. To derive ($b$), we use the probability generating functional (PGFL) of a PPP with density $\lambda_b$. ($c$) follows from distribution integral strategies and $\eta_{L} = \frac{r}{r_L}$.  

Similarly, the Laplace transform of useful signals can be given by 
\vspace{-1mm}
\begin{equation}
	{\rm{E}}\!\left[ {{e^{ - z U}}} \right] \!=\! \exp \!\bigg( \! - \pi \kappa_c \lambda_b {r^2}{\rm{H}}_1\left( { zp^c,M_{\mathrm{t}}-1,\alpha ,\eta_{L} } \right) \!\bigg),
	\vspace{-1mm}
\end{equation}
where ${\rm{H}}_1\left( {x,K,\alpha ,\eta_L } \right)  = \frac{1}{{{\eta ^2}}}\left( {1 - \frac{1}{{{{\left( {1 + x{\eta ^\alpha }} \right)}^K}}}} \right) + \frac{1}{{{{\left( {1 + x} \right)}^K}}} - 1 + K{x^{\frac{2}{\alpha }}}\!\left(\! {B\left( {\frac{x}{{x + 1}},1 - \frac{2}{\alpha },K + \frac{2}{\alpha }} \right) \! - \! B\left( {\frac{{x{\eta ^\alpha }}}{{x{\eta ^\alpha } + 1}},1 - \frac{2}{\alpha }, K + \frac{2}{\alpha }} \right)} \right)$.
Similarly, for the intra-cluster interference, i.e., the BSs declining the service requirement, the Laplace transform $I_1$ can be given by
\vspace{-1mm}
\begin{equation}\label{CommunicationEquationExpression2}
	\begin{aligned}
		{{\cal L}_{{I_{1}}}}(z) &= {\rm{E}}_{\Phi_b, g_i}\left[ {\exp \left( { - z{{\left\| {{{\bf{d}}_1}} \right\|}^\alpha }\sum\nolimits_{i \in \{\Phi_c \backslash \Phi_a \}}  {{{\left\| {{{\bf{d}}_i}} \right\|}^{ - \alpha }}} {g_j}} \right)} \right]  \\
		&{=} \exp \bigg(  - \pi (1-\kappa_c) \lambda_b {r^2}{\rm{H}}_1\left( {z,1,\alpha ,\eta_{L} } \right) \bigg) .
		\vspace{-1mm}
	\end{aligned}
\end{equation}
This thus completes the proof.

\section*{Appendix E: \textsc{Proof of Theorem \ref{CommunicationTightExpression}}}

According to (\ref{CommunicationBasicEquation}), the date rate under the conditional distance $r$ can be given by ${\rm{E}}\left[ {\log \left( {1 + {\rm{SI}}{{\rm{R}}_c}} \right)\mid r} \right] 
= \int_0^\infty  {\frac{{1 - {\rm{E}}\left[ {{e^{ - z{P_c}{g_1}}}} \right]{\rm{E}}\left[ {{e^{ - z U }}} \right]}}{z}} {\rm{E}}\left[ {{e^{ - z I_1 }}} \right]{\rm{E}}\left[ {{e^{ - z I_2 }}} \right]{\rm{d}}z$.
Then, the conditional expected spectrum efficiency can be given by
\vspace{-1mm}
\begin{equation}\label{longEquation2}
	\begin{aligned}
		&\int_0^\infty \!\! \int_0^1 \!\! \int_0^\infty \!\! \frac{{1 \! - \! {\rm{E}}\left[ {{e^{ - z{g_1}}}} \right]\exp \left( { - \pi \kappa_c \lambda_b {r^2}{{\rm{H}}_1}\left( { zp^c,M_{\mathrm{t}}-1,\alpha ,\eta_{L} } \right)} \right)}}{z}  \\
		&\times \exp \left( { - \pi \left( {1 - \kappa_c } \right)\lambda_b {r^2}{{\rm{H}}_1}\left( {z,1,\alpha ,\eta_{L} } \right)} \right)\exp \big(  - \pi \lambda_b {r^2}  \\
		& {{\rm{H}}_2}\left( {z,\alpha ,\eta_{L} } \right)\big) {f_{\eta_L}}\left( \eta  \right) {f_r}\left( r \right) {\rm{d}} r  {\rm{d}}\eta  {\rm{d}}z.
		\vspace{-1mm}
	\end{aligned}
\end{equation}
where ${f_r}\left( r \right) = 2\pi {\lambda _b}r{e^{ - \pi {\lambda _b}{r^2}}}$.
According to Lemma 3 in \cite{Zhang2014StochasticGeometry}, the PDF of the distance ratio $\eta_L$ can be given by 
\begin{equation}\label{RatioEquation}
	{f_{\eta_L}}\left( x \right) = 2\left( {L - 1} \right)x{\left( {1 - {x^2}} \right)^{L - 2}}.
\end{equation}
Then, by plugging the Laplace transforms of useful signal and interference in Lemma \ref{LaplaceTransform} to ${\rm{E}}\left[ {\log \left( {1 + {\rm{SI}}{{\rm{R}}_c}} \right)\mid r} \right]$, the equation (\ref{CommunicationTightExpression}) can be obtained. 
This thus completes the proof.

\footnotesize  	
\bibliography{mybibfile}
\bibliographystyle{IEEEtran}
\normalsize
	
\end{document}